\documentclass[11pt]{article}

\usepackage{microtype}
\usepackage{fullpage}

\usepackage{comment}
\usepackage{enumerate}

\usepackage{makeidx}  

\usepackage{amsthm,amsmath,amstext,amssymb}
\usepackage{xspace}
\usepackage{tikz}
 \usepackage{subcaption}
\newcommand{\F}{\mathcal{F}}
\newcommand{\Aa}{\mathcal{A}}

\newtheorem{theorem}{Theorem}
\newtheorem{proposition}[theorem]{Proposition}

\newtheorem{hypothesis}{Hypothesis}
\newtheorem{lemma}[theorem]{Lemma}
  \newtheorem{mclaim}{Claim}
  \newtheorem{mfact}{Fact}
  \newtheorem{step}{Step}

\newcommand{\twotable}{{\sc{$2$-\textsc{Table}}}\xspace}

\newcommand{\cO}{\mathcal{O}}
\newcommand{\cOs}{\mathcal{O}^*}
\newcommand{\sm}{\setminus}
\newcommand{\eps}{\epsilon}
\newcommand{\gclass}{\Pi}
\newcommand{\cbound}{\aleph}
\newcommand{\alg}{\mathcal{A}}

\newcommand{\cqed}{\renewcommand{\qedsymbol}{$\lrcorner$}}

\begin{document}

\title{~Largest chordal and interval subgraphs faster than $2^n$~\footnote{The research leading to these results has received funding from the European Research Council under the European Union's Seventh Framework Programme (FP/2007-2013) / ERC Grant Agreement n. 267959  and the Research Council of Norway (Yggdrasil mobility programme 2013-2014, project number 227328/F11), as well as from the Ministry of Education and Science of the Russian Federation, project 8216.}}

\author{
  Ivan Bliznets
  \thanks{
    St. Petersburg Academic University of the Russian Academy of Sciences, Russia,
      \texttt{ivanbliznets@tut.by}.
  }
  \and
  Fedor V. Fomin
  \thanks{
    Department of Informatics, University of Bergen, Norway, \texttt{\{fomin|michal.pilipczuk|yngvev\}@ii.uib.no}.
  }\addtocounter{footnote}{-1}
  \and
  Micha\l{} Pilipczuk\footnotemark\addtocounter{footnote}{-1}
  \and
  Yngve Villanger\footnotemark\addtocounter{footnote}{-1}
  }

\date{}

\maketitle

\begin{abstract}
We prove that in an $n$-vertex graph,    induced chordal and interval  subgraphs with the  maximum number of vertices can be found in time $\cO(2^{\lambda n})$ for some $\lambda<1$.  These are the first algorithms breaking the trivial $2^n n^{\cO(1)}$ bound of the brute-force search for these problems.

\end{abstract}

\section{Introduction}

The area of exact exponential algorithms is about   solving  intractable  problems  faster than the trivial exhaustive search, though still in exponential time  \cite{FominKratschbook10}.  
In this paper, we give  algorithms computing maximum induced chordal and interval subgraphs in a graph faster than the trivial brute-force search. These problems are interesting cases of a more general meta-problem \textsc{Maximum Induced $\gclass$-Subgraph}, where for a given graph $G$ and hereditary\footnote{A class of graphs is {\em{hereditary}} if it is closed under taking induced subgraphs.} class of graphs $\gclass$ one asks for an induced subgraph belonging $\gclass$ with the maximum possible number of vertices.

By the result of Lewis  and Yannakakis \cite{LewisY80}, the \textsc{Maximum Induced $\gclass$-Subgraph} problem is NP-hard for every non-trivial class of graphs $\gclass$. 
Different classes $\gclass$ were studied in the literature; examples include classes of graphs that are edgeless, planar, outerplanar, bipartite, complete bipartite, acyclic, degree-constrained, chordal etc. From the point of view of exact algorithms, as far as membership in $\gclass$ can be tested in polynomial time, a trivial brute-force search trying all possible vertex subsets of $G$ solves  
\textsc{Maximum Induced $\gclass$-Subgraph} in time $\cOs(2^n)$ on an $n$-vertex graph $G$.\footnote{The $\cOs(\cdot)$ notation suppresses terms polynomial in the input size.}  
However, many algorithms for \textsc{Maximum Induced $\gclass$-Subgraph} which are faster than $\cOs(2^n)$ can be found in the literature for explicit classes $\gclass$. 
Notable examples are $\gclass$ being the class of graphs that are edgeless~\cite{Robson86} (equivalent to \textsc{Maximum Independent Set}), acyclic~\cite{FominGPR08-On} (equivalent to \textsc{Maximum Induced Forest}), bipartite~\cite{RamanSS07}, planar~\cite{FominTV11}, degenerate~\cite{PilipczukP12}, regular \cite{GuptaRS12},  cluster graphs~\cite{FominGKLS10}, bounded treewidth \cite{Fomin:2010ys}, or bicliques~\cite{GaspersKL12}, see Table~\ref{table1}. Very recently, Fomin et al.~\cite{FominTV13} have shown that for every hereditary class of graphs $\gclass$ that have constant treewidth and are definable in counting monadic second-order logic (CMSO), the \textsc{Maximum Induced $\gclass$-Subgraph} problem can be solved in $\cO(1.7347^n)$ time.

The listed examples suggest that existence of algorithms faster than $2^n$ for \textsc{Maximum Induced $\gclass$-Subgraph} can be a phenomenon of a much more general nature. In fact, so far we do not know any efficiently recognizable class $\gclass$ for which a lower bound could be derived. Therefore,   the following   is a folklore open problem in the field; we could not find it  explicitly stated in the literature, so we state it below.

\begin{hypothesis}\label{conj:main}
For every hereditary graph class $\gclass$ that can be recognized in polynomial time, the \textsc{Maximum Induced $\gclass$-Subgraph} problem can be solved in time $\cOs(2^{\lambda n})$ for some $\lambda<1$.
\end{hypothesis}

We expect that some additional weak conditions on~$\gclass$ may be needed to provide a positive answer to 
   hypothesis we discuss propositions of such relaxations in Section~\ref{sec:conlusion}. Thus, the aforementioned results~\cite{FominGKLS10,FominGPR08-On,FominTV11,FominTV13,GaspersKL12,PilipczukP12,RamanSS07,Robson86} can be viewed as partial progress on Hypothesis~\ref{conj:main}: the idea is to investigate how different features of the class $\gclass$ can be used to design an algorithm breaking the $2^n$ barrier.


\begin{table}[h!]
\begin{center}
\begin{tabular}{|c|c|c|}
\hline
Property & Time complexity &  Reference\\
\hline
edgeless & $\cO(1.2109^n)$  & Robson~\cite{Robson86}\\
acyclic & $\cO(1.7548^n)$ & Fomin et al.~\cite{FominGPR08-On}\\
bipartite & $\cO(1.62^n)$ & Raman et al.~\cite{RamanSS07}\\
planar & $\cO(1.7347^n)$ & Fomin et al.~\cite{FominTV11}\\
$d$-degenerate & $\cO((2-\eps_d)^n)$ & Pilipczuk$\times 2$~\cite{PilipczukP12}\\
cluster graph & $\cO(1.6181^n)$ & Fomin et al.~\cite{FominGKLS10}\\
biclique & $\cO(1.3642^n)$  & Gaspers et al.~\cite{GaspersKL12}\\
$o(n/\log{n})$ treewidth &  $\cO(1.7347^n)$  & Fomin and Villanger~\cite{Fomin:2010ys}\\
$r$-regular & $\cO((2-\eps_r)^n)$ & Gupta et al. ~\cite{GuptaRS12}\\
matching & $\cO(1.6957^n)$ & Gupta et al. ~\cite{GuptaRS12}\\
\hline
\end{tabular}
\end{center}
\caption{Known results for \textsc{Maximum Induced $\gclass$-Subgraph}}\label{table1}
\end{table}

Since every hereditary class of graphs $\gclass$ can be characterized by a (not necessarily finite) set of forbidden induced subgraphs, there is an equivalent formulation of the \textsc{Maximum Induced   $\gclass$-Subgraph} problem. For a set of graphs $\F$, a graph $G$ is called $\F$-free if it contains no graph from $\F$ as an induced subgraph. The \textsc{Maximum $\F$-free Subgraph} problem is to find a maximum induced $\F$-free subgraph of $G$. Clearly, if $\F$ is the set of forbidden induced subgraphs for $\gclass$, then the \textsc{Maximum Induced $\gclass$-Subgraph} problem and the \textsc{Maximum $\F$-free Subgraph} problem are equivalent.

It is well known that when the set $\F$ is \emph{finite}, then \textsc{Maximum $\F$-free Subgraph} can be solved in time $\cOs(2^{\lambda n})$, where $\lambda <1$. This can be seen by applying a simple branching arguments, see Proposition~\ref{lemma:finite_deletion}, or by reducing to the \textsc{$d$-Hitting Set} problem, which is solvable faster than $\cOs(2^n)$ for every fixed $d$~\cite{FominGKLS10,Gaspers:2008rf}. Examples of $\F$-free classes of graphs for some finite set $\F$ are split  graphs, cographs, line graphs or trivially perfect graphs; see the book \cite{brandstadt1999graph} for more information on these graph classes.

It is however completely unclear if anything faster than the trivial brute-force is possible in the case when $\F$ is an infinite set, even when $\F$ consists of very simple graphs. One of the most known  and well studied classes of $\F$-free graphs is the class of  \emph{chordal graphs}, where $\F$ is the set of all cycles of length more than three.  Chordal graphs form a fundamental class of graphs whose properties are well understood. Another fundamental class of graphs is the class of \emph{interval graphs}. We refer to the book  of Golumbic  for an overview of properties and applications of chordal and interval graphs \cite{Golumbic80}. In spite of nice structural properties of these graphs,  no exact algorithms for \textsc{Maximum Induced Chordal Subgraph} and \textsc{Maximum Induced Interval Subgraph} problems better than the trivial $\cOs(2^n)$ were known prior to our work.
 
\smallskip\noindent\textbf{Our results.} 
We define four properties of a graph class and give an algorithm that, for every fixed graph class $\gclass$ (not part of the input) satisfying these properties, and for a given $n$-vertex graph $G$, finds a maximum induced subgraph of $G$ belonging to $\gclass$ in time $\cOs(2^{\lambda n})$ for some $\lambda<1$, where $\lambda$ depends only on the class $\gclass$. Because classes of chordal and interval graphs satisfy the required properties, as an immediate corollary we obtain that \textsc{Maximum Induced Chordal Subgraph} and \textsc{Maximum Induced Interval Subgraph} can be solved in time $\cOs(2^{\lambda n})$ for some $\lambda<1$.

When pipelined with simple branching arguments, our algorithms can be used to obtain time $\cOs(2^{\lambda n})$  algorithms for some $\lambda<1$  for a variety of other \textsc{Maximum Induced $\gclass$-Subgraph} problems, where class $\gclass$ comprises chordal/interval graphs that moreover contain no induced subgraph from a \emph{finite} forbidden set of graphs. 
Examples of such classes are proper interval graphs, Ptolemaic graphs, block graphs, or proper circular-arc graphs; see \cite{brandstadt1999graph} for definitions and discussions of these graph classes.

The main intention of our work was to break the trivial $2^n$ barrier, and thus to provide a new insight into Hypothesis~\ref{conj:main} by analyzing chordal and chordal-like graph classes. For this reason, we did not try to optimize the constant $\lambda$ in the exponent. There are several places where the running time of our algorithm can be improved at a cost of more involved arguments or intensive case analyses. However, we would like to stress again that the main motivation of our work is the theoretical study of Hypothesis~\ref{conj:main}, rather than pursuit of really efficient algorithms for the respective problems. Therefore, we refrain from giving these improvements and prefer keeping the arguments as simple as possible.

\smallskip\noindent\textbf{Organization.} In Section~\ref{sec:preliminaries} we give notation and recall known tools that will be used later. In Section~\ref{sec:properties} we discuss the four properties of a graph class that are needed for our algorithm to be applicable. Section~\ref{sec:algo} contains the description of the algorithm. For the convenience of the reader, in Section~\ref{app:choice} we summarize the order of choice of small constants used by the algorithm. Finally, in Section~\ref{sec:conlusion} we give some concluding remarks.

\section{Preliminaries}\label{sec:preliminaries}

\noindent{\bf{Standard graph notation.}} We denote by $G=(V,E)$ a finite, undirected and simple graph with vertex set $V(G)=V$ and edge set $E(G)=E$. We also use $n$ to denote the number of vertices in $G$. For a nonempty subset of vertices $W\subseteq V$, a subgraph {\em{induced}} by $W$ is defined as $G[W]=(W,E\cap (W\times W))$. An \emph{induced subgraph} of a graph is a subgraph induced by some subset of vertices. A \emph{clique} is a set of vertices inducing a complete subgraph. We say that a vertex set $W\subseteq V$ is \emph{connected} if $G[W]$ is connected. The \emph{open neighborhood} of a vertex $v$ is $N(v)=\{u\in V:~uv \in E\}$ and the \emph{closed neighborhood} is $N[v] = N(v) \cup \{v\}$. For a vertex set $W\subseteq V$ we put  $N(W) = \bigcup_{v \in W} N(v)\sm W$ and $N[W] = N(W) \cup W$. Whenever the graph to which the neighbourhood operation is applied is not clear from the context, we put it in the subscript of $N$. For a vertex subset $X$ of a graph $G$, we use $G\setminus X$ to denote the subgraph of $G$ induced by $V(G)\setminus X$.

A {\emph{graph class}} $\gclass$ is simply a family of graphs. We often use terms $\gclass$-graph or $\gclass$-subgraph to express membership in $\gclass$. We say that a graph class is {\emph{hereditary}} if $\gclass$ is closed under taking induced subgraphs. Every hereditary graph class can be described by a (possibly infinite) list of minimum forbidden induced subgraphs $\F_\gclass$: graph $G$ is in $\gclass$ if and only if it does not contain any induced subgraph from $\F_\gclass$, and for each $H\in \F_\gclass$ every induced subgraph of $H$, apart from $H$ itself, belongs to $\gclass$. The class of graphs not containing any induced subgraph from a list $\F$ will be denoted by {\emph{$\F$-free graphs}}.

\textit{Chordal graphs} are graphs not containing any induced cycles of length more than three, that is, chordal graphs are $\F$-free graphs where the set $\F$ consists of all cycles of length more than three. Chordal graphs are hereditary and polynomial-time recognizable~\cite{Golumbic80}. Chordal graphs admit many more characterizations, for example they are exactly graphs admitting a decomposition into a clique tree. A useful corollary of this fact is the following folklore lemma.

\begin{proposition}[Folklore]\label{le:balanced-separator}
If $H$ is a chordal graph, then there exists a clique $S$ in $H$ and a partition of $V(H)\setminus S$ into two subsets $X_1,X_2$, such that 
\begin{itemize}
\item[(i)] $|X_1|,|X_2|\leq \frac{2}{3}|V(H)|$, and \item[(ii)] there is no edge between $X_1$ and $X_2$.
\end{itemize}
\end{proposition}

Such a set $S$ is called a {\emph{$\frac{2}{3}$-balanced clique separator}} in $H$. Note that since $|X_2|\leq \frac{2}{3}|V(H)|$, then $|X_1|=|V(H)|-|S|-|X_1|\geq \frac{1}{3}|V(H)|-|S|$, and the same holds also for $X_2$.

Interval graphs form a subclass of chordal graphs admitting a decomposition into a clique path instead of less restrictive clique tree. Interval graphs are also hereditary and polynomial-time recognizable~\cite{Golumbic80}. Their characterization in terms of minimal forbidden induced subgraphs was given by Lekkerkerker and Boland \cite{LekkerkerkerB62}; see Figure~\ref{fig:interval} for reference. The book of Golumbic~\cite{Golumbic80} provides a thorough introduction to chordal and interval graphs. 

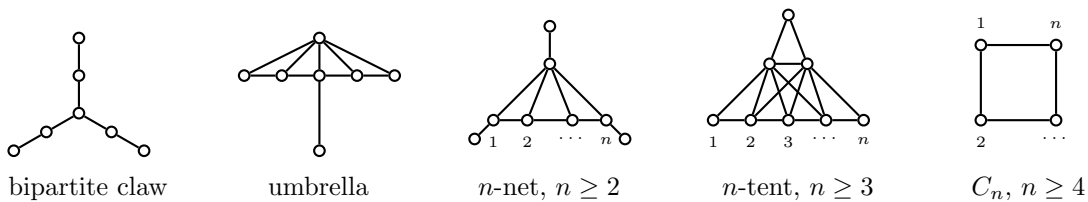
\begin{figure}[h]
  \centering

\begin{subfigure}[t]{0.13\textwidth}
\begin{tikzpicture}[thick,scale=.5]
\node[draw,circle,fill=white,minimum size=4pt,
                            inner sep=0pt] (c) at (0:0) {};
\node[draw,circle,fill=white,minimum size=4pt,
                            inner sep=0pt] (r1) at (-30:1) {} edge[-](c) ;
\node[draw,circle,fill=white,minimum size=4pt,
                            inner sep=0pt] (r2) at (-30:2) {} edge[-](r1);
\node[draw,circle,fill=white,minimum size=4pt,
                            inner sep=0pt] (u1) at (90:1) {} edge[-](c);
\node[draw,circle,fill=white,minimum size=4pt,
                            inner sep=0pt] (u2) at (90:2) {} edge[-](u1);
\node[draw,circle,fill=white,minimum size=4pt,
                            inner sep=0pt] (l1) at (210:1) {} edge[-](c);
\node[draw,circle,fill=white,minimum size=4pt,
                            inner sep=0pt] (l2) at (210:2) {} edge[-](l1);
\end{tikzpicture}
\caption{bipartite claw}
\end{subfigure}
\hspace{.04\textwidth}
\begin{subfigure}[t]{0.13\textwidth}
\begin{tikzpicture}[thick,scale=0.5]
\node[draw,circle,fill=white,minimum size=4pt,
                            inner sep=0pt] (c) at (0,1) {};
\node[draw,circle,fill=white,minimum size=4pt,
                            inner sep=0pt] (u) at (0,2) {} edge[-](c) ;
\node[draw,circle,fill=white,minimum size=4pt,
                            inner sep=0pt] (r1) at (1,1) {} edge[-](c) edge[-](u);
\node[draw,circle,fill=white,minimum size=4pt,
                            inner sep=0pt] (r2) at (2,1) {} edge[-](u) edge[-](r1);
\node[draw,circle,fill=white,minimum size=4pt,
                            inner sep=0pt] (l1) at (-1,1) {} edge[-](c) edge[-](u);
\node[draw,circle,fill=white,minimum size=4pt,
                            inner sep=0pt] (r1) at (-2,1) {} edge[-](u) edge[-](l1);
\node[draw, circle,fill=white,minimum size=4pt,inner sep=0pt] (b) at (0,-1){} edge[-](c);
\end{tikzpicture}
\caption{umbrella}
\end{subfigure}
\hspace{.04\textwidth}
\begin{subfigure}[t]{0.13\textwidth}
\begin{tikzpicture}[thick,scale=0.5]
\node[draw,circle,fill=white,minimum size=4pt,
                            inner sep=0pt] (u2) at (0,2) {};
\node[draw,circle,fill=white,minimum size=4pt,
                            inner sep=0pt] (u1) at (0,1) {} edge[-](u2) ;
\node[draw,circle,fill=white,minimum size=4pt,
                            inner sep=0pt,label=below:{\tiny{$1$}}] (l2) at (-1.5,-0.5) {} edge[-](u1);
\node[draw,circle,fill=white,minimum size=4pt,
                            inner sep=0pt, label=below:{\tiny{$2$}}] (l1) at (-0.6,-0.5) {} edge[-](u1) edge[-](l2);
                            
\node[draw,circle,fill=white,minimum size=4pt,
                            inner sep=0pt, label=below:{\tiny{$n$}}] (r2) at (1.5,-0.5) {} edge[-](u1)  ;
\node[draw,circle,fill=white,minimum size=4pt,
                            inner sep=0pt, label=below:{\tiny{$\dots$}}] (r1) at (0.6,-0.5) {} edge[-](u1) edge[-](r2) edge[-](l1);
                            
\node[draw,circle,fill=white,minimum size=4pt,
                            inner sep=0pt] (ll) at (-2,-1) {} edge[-](l2);
\node[draw,circle,fill=white,minimum size=4pt,
                            inner sep=0pt] (rr) at (2,-1) {} edge[-](r2);

\end{tikzpicture}
\caption{$n$-net, $n\geq 2$}
\end{subfigure}
\hspace{.04\textwidth}
\begin{subfigure}[t]{0.16\textwidth}
\begin{tikzpicture}[thick,scale=0.5]
\node[draw,circle,fill=white,minimum size=4pt,
                            inner sep=0pt] (u) at (0,2) {};
                            
\node[draw,circle,fill=white,minimum size=4pt,
                            inner sep=0pt] (l) at (-0.5,0.7) {} edge[-](u) ;
\node[draw,circle,fill=white,minimum size=4pt,
                            inner sep=0pt] (r) at (0.5,0.7) {} edge[-](u) edge[-](l);
                            
\node[draw,circle,fill=white,minimum size=4pt,
                            inner sep=0pt,label=below:{\tiny{$1$}} ] (b1) at (-2,-0.8) {} edge[-](l);
\node[draw,circle,fill=white,minimum size=4pt,
                            inner sep=0pt, label=below:{\tiny{$2$}}] (b2) at (-1,-0.8) {} edge[-](b1) edge[-](l) edge[-](r) ;
\node[draw,circle,fill=white,minimum size=4pt,
                            inner sep=0pt, label=below:{\tiny{$3$}}] (b3) at (0,-0.8) {} edge[-](b2) edge[-](l) edge[-](r);
\node[draw,circle,fill=white,minimum size=4pt,
                            inner sep=0pt,label=below:{\tiny{$\dots$}}] (b4) at (1,-0.8) {} edge[-](b3) edge[-](l) edge[-](r);
\node[draw,circle,fill=white,minimum size=4pt,
                            inner sep=0pt,label=below:{\tiny{$n$}}] (b5) at (2,-0.8) {} edge[-](r) edge[-](b4);

\end{tikzpicture}
\caption{$n$-tent, $n\geq 3$}
\end{subfigure}
\hspace{.04\textwidth}
\begin{subfigure}[t]{0.10\textwidth}
\begin{tikzpicture}[thick,scale=0.5]
\node[draw,circle,fill=white,minimum size=4pt,
                            inner sep=0pt, label=above:{\tiny{$n$}}] (ur) at (1,1.5) {};
\node[draw,circle,fill=white,minimum size=4pt,
                            inner sep=0pt,label=above:{\tiny{$1$}}] (ul) at (-1,1.5) {} edge[-](ur) ;
                            
\node[draw,circle,fill=white,minimum size=4pt,
                            inner sep=0pt,label=below:{\tiny{$2$}}] (bl) at (-1,-0.5) {} edge[-](ul);
\node[draw,circle,fill=white,minimum size=4pt,
                            inner sep=0pt, label=below:{\tiny{$\dots$}}] (br) at (1,-0.5) {} edge[-](bl) edge[-](ur);

\end{tikzpicture}
\caption{$C_n$, $n\geq 4$}
\end{subfigure}

\caption{Forbidden induced subgraphs for interval graphs}\label{fig:interval}
\end{figure}

We now describe the classical tools needed for the algorithm. The following result basically follows from the observation that branching on forbidden structures of constant size always leads to complexity better than $2^n$.

\begin{proposition}\label{lemma:finite_deletion}
Let $\F$ be a finite set of graphs and let $\ell$ be the maximum number of vertices in a graph from $\F$. Let $\Pi$ be a hereditary graph class that is polynomial-time recognizable. Assume that there exists an algorithm $\Aa$ that for a given $\F$-free graph $G$ on $n$ vertices, in $\cOs(2^{\eps n})$ time finds a maximum induced $\Pi$-subgraph of $G$, for some $\eps<1$. Then there exists an algorithm $\Aa'$ that for a given  graph $G$ on $n$ vertices, finds a maximum induced $\F$-free $\Pi$-graph in $G$ in time  $\cOs(2^{\eps' n})$, where $\eps'<1$ is a constant  depending on $\eps$ and $\ell$.
\end{proposition}

\begin{proof}
Let $\Pi'$ be the class of $\F$-free $\Pi$-graphs; note that for constant $\ell$, $\Pi'$ is polynomial-time recognizable. Algorithm $\Aa'$, given an $n$-vertex graph $G=(V,E)$, tries to find a maximum induced $\Pi'$-subgraph using standard  branching arguments. At each point the algorithm maintains two disjoint sets $A,D\subseteq V$; at the starting point $A=D=\emptyset$. Given $A,D$, the algorithm tries to find a maximum size set $X$ inducing a $\Pi'$-graph such that $A\subseteq X$ and $D\cap X=\emptyset$. Whenever we arrive at a situation when $|A\cup D|>(1-\eps)n$, we stop the branching procedure and perform a brute force check on the remaining vertices of $V\setminus (A\cup D)$. That is, we examine all subsets $A'\subseteq V\setminus (A\cup D)$ and test whether $G[A\cup A']$ induces a $\Pi'$-graph. This takes time $\cOs(2^{|V\setminus (A\cup D)|})\leq \cOs(2^{\eps n})$.

At each step of the branching procedure we check in polynomial time whether $G\setminus D$ contains a subgraph isomorphic to any graph of $\F$. Assume first that we have found such a subgraph and let $S\subseteq V\setminus D$ be its vertex set. Clearly, for every induced $\Pi'$-subgraph, at least one of vertices of $S$ is not contained in this subgraph. As vertices of $S\cap A$ have to be in the solution searched in this branch, we branch on set $S\setminus A$. More precisely, for every partition $(A',D')$ of $S\setminus A$ where $D'$ is nonempty, we produce a branch in which $A'$ is incorporated into $A$ and $D'$ is incorporated into $D$. Note that this leads to $2^{\ell'}-1$ branches produced and increasing $|A\cup D|$ by $\ell'$, where $\ell'=|S\setminus A|\leq \ell$. Note moreover that since $\ell'\leq \ell$, then $2^{\ell'}-1\leq 2^{\eps_\ell \ell'}$ for some $\eps_\ell<1$ depending on $\ell$.

Assume now that $G\setminus D$ contains no induced subgraph from $\F$, hence it is $\F$-free. We apply the algorithm $\Aa$ to $G\setminus D$ to compute the maximum induced $\Pi$-subgraph of $G\setminus D$. As $G\setminus D$ is $\F$-free, this subgraph is in fact in the class $\Pi'$. Note here that at this point we relax the condition that the set we are looking for has to contain $A$ as a subset, however this does not affect correctness of the algorithm: the found subgraph is still an induced $\Pi'$-subgraph of $G$, so it can be only larger than the solution we are looking for in this branch. The running time of the application of $\Aa$ is $\cOs(2^{\eps|V\setminus D|})\leq \cOs(2^{\eps n})$.

Let us now discuss the running time of the algorithm. Note that at the point of applying brute-force check we have that $(1-\eps) n+\ell\geq |A\cup D|>(1-\eps) n$, as $|A\cup D|$ can increase by at most $\ell$ at each step. Each branching step increases $|A\cup D|$ by some $\ell'$ and introduces at most $2^{\eps_\ell \ell'}$ subbranches, hence the total number of instances where algorithm $\Aa$ or a brute-force check is applied is at most $2^{\eps_\ell((1-\eps)n+\ell)}=\cO(2^{\eps_\ell (1-\eps) n})$. Application of brute-force or algorithm $\Aa$ takes $\cOs(2^{\eps n})$ time. Hence, the total running time is $\cOs(2^{\eps' n})$, where $\eps'=\eps_\ell(1-\eps)+\eps<1$. 
\end{proof}

The following proposition from \cite{FominV12} will be useful for us to guess connected sets of vertices with small running-time overhead.
\begin{proposition}[\cite{FominV12}]\label{le:connectedComp} Let $G=(V,E)$ be a graph. For every $v\in
V$, and $b,f\geq 0$,  the number of connected vertex subsets
$B\subseteq V$ such that
\begin{itemize}
\item[(i)]
 $v \in B$,\item[(ii)] $|B| = b+1$, and \item[(iii)] $|N(B)|=f$,  \end{itemize}
  is at most $\binom{b+f}{b}$.
 Moreover, all such subsets can be enumerated in time $\cOs(\binom{b+f}{b})$.
\end{proposition}

The last necessary ingredient is the 
classical idea used by Schroeppel and Shamir \cite{SchroeppelS81-A} for  solving  \textsc{Subset Sum} by reducing it to an instance of $2$-\textsc{Table}.
 In the  $2$-\textsc{Table} problem, we are given two $k\times m_i$ matrices $T_i$, $i=1,2$, and a vector $\vec{s}\in \mathbb{Q}^{k}$. Columns of each matrix are $m_i$ vectors of $\mathbb{Q}^{k}$. The question is, if there is a column of the first matrix and a column of the second matrix such that the sum of these two columns is equal to $\vec{s}$. A trivial solution to the \twotable problem would be to try all possible pairs of 
vectors; however, this problem can be solved more efficiently. We can sort columns of $T_1$ lexicographically in $\cO(km_1\log m_1)$ time, and for every column $\vec{v}$ of $T_2$ check whether $T_1$ contains a column equal to $\vec{s}-\vec{v}$ in $\cO(k\log m_1)$ time using binary search.
\begin{proposition}[\cite{SchroeppelS81-A}]\label{lem:2table}
The \twotable problem can be solved in time $\cO((m_1+m_2) k\log m_1)$.
\end{proposition}

\section{Properties of the graph class}\label{sec:properties}

In this section we gather the required properties of the graph class $\gclass$ for our algorithm to be applicable. We consider only hereditary subclasses of chordal graphs, hence our first property is the following.

\vskip 0.1cm
\noindent{\bf{Property (1).}} $\gclass$ is a hereditary subclass of chordal graphs.
\vskip 0.1cm

As $\gclass$ is hereditary, it may be described by a list of vertex-minimal forbidden induced subgraphs $\F_\gclass$. We need the following properties of $\F_\gclass$:

\vskip 0.1cm
\noindent{\bf{Property (2).}} All graphs in $\F_\gclass$ are connected, and all of them do not contain a clique of size $\cbound+1$ for some universal constant $\cbound$.
\vskip 0.1cm

For chordal graphs $\F_\gclass$ consists of cycles of length at least $4$, hence $\cbound=2$. For interval graphs, an inspection of the list of forbidden induced subgraphs, depicted on Figure~\ref{fig:interval}, shows that we may take $\cbound=4$. In the following, we always treat $\cbound$ as a universal constant for class $\gclass$ on which all the later constants may depend; moreover, $\cbound$ may influence the exponents of polynomial factors hidden in the $\cOs$ notation. Let us remark that connectedness of all the forbidden induced subgraphs is equivalent to requiring $\gclass$ to be closed under taking disjoint union.
An example of a subclass of chordal graphs not satisfying this property, is the class of strongly  chordal graphs. The reason for that is that minimal forbidden subgraphs of strongly chordal graphs can contain  a clique of any size, see \cite{brandstadt1999graph} for more information on this class of graphs.  

Thirdly, we need our graph class to be efficiently recognizable.

\vskip 0.1cm
\noindent{\bf{Property (3).}} $\gclass$ is polynomial-time recognizable.
\vskip 0.1cm

 Chordal graphs and interval graphs have polynomial time recognition algorithms, see e.g. \cite{Golumbic80}. For our arguments to work we need one more algorithmic property. The property that we need can be described intuitively as robustness with respect to clique separators.  More precisely, we need the following statement.

\vskip 0.1cm
\noindent{\bf{Property (4).}} There exists a polynomial-time algorithm $\alg$ that takes as input a graph $G$ together with a clique $S$ in $G$. The algorithm answers YES or NO, such that the following conditions are satisfied:
\begin{itemize}
\item If $\alg$ answers YES on inputs $(G_1,S_1)$ and $(G_2,S_2)$ where $|S_1|=|S_2|$, then graph $G'$, obtained by taking disjoint union of $G_1$ and $G_2$ and identifying every vertex of $S_1$ with a different vertex of $S_2$ in any manner, belongs to~$\gclass$.
\item If $G\in\gclass$, then there exists a clique separator $S$ in $G$ such that $V(G)\setminus S$ may be partitioned into two sets $X_1,X_2$ such that (i) $|X_1|,|X_2|\leq \frac{2}{3}|V(G)|$, (ii) there is no edge between $X_1$ and $X_2$, (iii) $\alg$ answers YES on $(G[X_1\cup S],S)$ and on $(G[X_2\cup S],S)$.
\end{itemize}
\vskip 0.1cm

Observe that Property (1) and Proposition~\ref{le:balanced-separator} already provides us with some $\frac{2}{3}$-balanced clique separator $S$ of $G$. Shortly speaking, Property (4) requires that in addition belonging to $\gclass$ may be tested by looking at $G[X_1\cup S]$ and $G[X_2\cup S]$ independently. For chordal graphs, Property (4) follows from Proposition~\ref{le:balanced-separator} and a folklore observation that if $S$ is a clique separator in a graph $G$, with $(X_1,X_2)$ being a partition of $V(G)\setminus S$ such that there is no edge between $X_1$ and $X_2$, then $G$ is chordal if and only if $G[X_1\cup S]$ and  $G[X_2\cup S]$ are chordal. Hence, we may take chordality testing for the algorithm $\alg$.

For interval graphs, let us take the clique path of $G$ and examine a clique separator $S$ such that there is at most half of vertices before it and at most half after it. Let $X_1$ be the vertices before $S$ on the clique path, and $X_2$ be the vertices after $S$. Clearly, $S$ is then even a $\frac{1}{2}$-balanced clique separator, with partition $(X_1,X_2)$ of $V(G)\setminus S$. Then it follows that $G[X_1\cup S]$ and $G[X_2\cup S]$ admit clique paths in which $S$ is one of the end bags of the path. On the other hand, assume that we are given any two graphs $G_1,G_2$ with equally sized cliques $S_1,S_2$, such that $G_1,G_2$ admit clique paths with $S_1$, $S_2$ as the end bags. Then we may create a clique path of the graph $G'$ obtained from the disjoint union of $G_1$ and $G_2$ and identification of $S_1$ and $S_2$ in any manner, by simply taking the clique paths for $G_1$ and $G_2$ and identifying the end bags containing $S_1$ and $S_2$, respectively. Hence, as $\alg$ we may take an algorithm which for input $(G,S)$ checks whether $G$ is interval and admits a clique path with $S$ as the end bag. Such a test may be easily done as follows: we add two vertices $v$, $v'$ to $G$, where $v$ is adjacent to $v'$ and to every vertex of $S$, while $v'$ is adjacent only to $v$. In this manner we force $S$ to be the end bag, and run the intervality test. Hence, interval graphs also satisfy Property (4).

\section{The algorithm}\label{sec:algo}
In this section we prove the main result of the paper, which is the following.

\begin{theorem}\label{thm:chord_2n}
If $\gclass$ satisfies Properties (1)-(4), then there exists an algorithm which, given an $n$-vertex graph $G$, returns a maximum induced subgraph of $G$ belonging to $\gclass$ in time $\cOs(2^{\lambda n})$ for some $\lambda<1$, where $\lambda$ depends only on $\cbound$.
\end{theorem}
As we already observed, chordal and interval graphs satisfy Properties (1)-(4). Thus Theorem~\ref{thm:chord_2n} implies immediately results claimed in the introduction. 
Our approach is based on a thorough investigation of the structure of a maximum induced subgraph. In each of the cases, we deploy a different strategy to identify possible suspects for an optimal solution. The properties we strongly rely on are the balanced separation property of chordal graphs (Property (4)), and conditions on minimal forbidden induced subgraphs for $\gclass$ (Property (2)).   

Let $G=(V,E)$. In the description of the algorithm we use several small positive constants: $\alpha,\beta,\gamma,\delta,\varepsilon$, and one large constant $L$. The final constant $\lambda$ depends on the choice of $\alpha,\beta,L,\gamma,\delta,\varepsilon$; during the description we make sure that constants $(\alpha,\beta,L,\gamma,\delta,\varepsilon)$ can be chosen so that $\lambda<1$. The choice of each constant depends on the later ones, e.g., having chosen $L,\gamma,\delta,\varepsilon$, we may find a positive upper bound on the value of $\beta$ so that we may choose any positive $\beta$ smaller than this upper bound. For reader's convenience, in Appendix~\ref{app:choice} we give a summary of the procedure of choosing constants.

Firstly, we observe that by Proposition~\ref{lemma:finite_deletion},  we may assume that the input graph does not contain any forbidden induced subgraph from $\F_\gclass$ of size at most $\ell$ for some constant $\ell$, to be determined later. Indeed, if we are able to find an algorithm for maximum induced $\gclass$-subgraph running in $\cOs(2^{\lambda n})$ time for some $\lambda<1$ and working in $\F'_\gclass$-free graphs, where $\F'_\gclass$ consists of graphs of $\F_\gclass$ of size at most $\ell$, then by Proposition~\ref{lemma:finite_deletion} we obtain an algorithm for maximum induced $\gclass$-subgraph working in general graphs and with running time $\cOs(2^{\lambda' n})$ for some $\lambda'<1$. Hence, from now on we assume that the input graph $G$ does not contain any forbidden induced subgraph from $\F_\gclass$ of size at most $\ell$.

The algorithm performs a number of {\emph{steps}}. After each step, depending on the result, the algorithm chooses one of the subcases.

\smallskip\noindent\textbf{Step 1.} \emph{Using the algorithm of Robson~\cite{Robson86}, in $\cOs(2^{0.276n})$ time find the largest clique $K$ in $G$.}\smallskip

We consider two cases: either $K$ is large enough to finish the search directly, or $K$ is small and we have a guarantee that the maximum induced $\gclass$-graph we are looking for contains only small cliques. The threshold for small/large is $\alpha n$ for a constant $\alpha > 0$, $\alpha < 1/48$, to be determined later.

\smallskip\noindent\textbf{Case~A:} \emph{$|K|\geq \alpha n$.}\smallskip

We show that in this case, the problem can be solved in $\cOs(2^{(1-(1-\kappa_0)\alpha)n})$ time for some $\kappa_0<1$ depending only on $\cbound$. We use the following auxiliary claim. 

\begin{lemma}\label{lem:chordal_plus_clique}
Let $P$ be a subset of vertices of an $n$-vertex graph $G$ that induces a graph belonging to $\gclass$, and let $K$ be a clique in $G$ such that 
$P\cap K=\emptyset$. Then in time $\cOs(2^{\kappa_0\cdot |K|})$ for some $\kappa_0<1$ depending only on $\cbound$ it is possible to find an induced subgraph of $G$ with the maximum number of vertices, where maximum is taken over all induced subgraphs $H$ of $G$ such that {\emph{(i)}} $H\in \gclass$, {\emph{(ii)}} $V(H)\setminus K=P$. In other words, the maximum is taken over all induced subgraphs belonging to $\gclass$ which can be obtained by adding some vertices of $K$ to $P$.  
\end{lemma}
\begin{proof}
For every nonempty subset $W$ of $K$ of size at most $\cbound$, we colour $W$ red if $G[W\cup P]\in \gclass$. Note that this construction may be performed using at most $\cbound\cdot|K|^\cbound$ tests of belonging to $\gclass$, hence in polynomial time for constant $\cbound$.

We observe that for every subset $X\subseteq K$, $G[P\cup X]$ belongs to $\gclass$ if and only if all nonempty subsets of $X$ of size at most $\cbound$ are red. Indeed, if the latter is not the case, there is a subset $W\subseteq X$ such that $G[P\cup W]\notin \gclass$, so by Property (1) $G[P\cup X]\notin \gclass$ as well. For the opposite direction, let us assume that $G[P\cup X]$ contains some forbidden induced subgraph $F\in \F_\gclass$. Then $|F\cap X|>\cbound$ because otherwise, by the definition of the colouring, $F\cap X$ would not be coloured red. But since $X$ is a clique, we conclude that $F$ contains a clique on $\cbound+1$ vertices, which is a contradiction with Property (2).

Hence, to obtain a maximum subgraph one has to find a maximum subset of $X$ such that all its subsets of size at most $\cbound$ are coloured red. This is equivalent to finding a maximum clique in a hypergraph with hyperedges of cardinality at most $\cbound$, which can be done using a  branching algorithm in $\cOs(2^{\kappa_0\cdot |K|})$ time for some $\kappa_0<1$, depending only on $\cbound$.

The branching algorithm maintains two disjoint sets of vertices $A,D$, at the beginning equal to $\emptyset$. Set $A$ consists of vertices that are guessed to be in the solution, while $D$ consists of vertices guessed to not be in the solution. The algorithm terminates the branch when $K\setminus D$ does not have any subset of size at most $\cbound$ not coloured red, and in this case $K\setminus D$ is produced as a candidate for the optimum $X$; the optimum $X$ is found as the largest set among the candidates. If the branch is not terminated, we infer that there must be a subset $W\subseteq K\setminus D$ of size at most $\cbound$ which is not coloured red. Clearly, at least one of the vertices of $W$ cannot be in the optimum $X$, hence we examine $W\setminus A$ and branch into $2^{|W\setminus A|}-1$ cases, in each fixing a different choice which vertices of $W\setminus A$ should go to $A$ and which should go to $D$; the omitted case is when all the considered vertices go to $A$. As $|W\setminus A|\leq \cbound$, we have that $2^{|W\setminus A|}-1\leq 2^{\kappa_0\cdot |W\setminus A|}$ for some $\kappa_0<1$ depending only on $\cbound$. Hence, we are able to fix alignment of $|W\setminus A|$ vertices by creating at most $2^{\kappa_0\cdot |W\setminus A|}$ branches, and the total running time $2^{\kappa_0\cdot |K|}$ follows.\end{proof}

We now do the following. Let $H$ be a maximum induced subgraph of $G$ belonging to $\gclass$. We branch into at most $2^{|V\setminus K|}$ subcases, in each fixing a different subset $P$ of $V\setminus K$ as $V(H)\setminus K$; we discard all the branches where the subgraph induced by $P$ does not belong to $\gclass$. For each branch, we use Lemma~\ref{lem:chordal_plus_clique} to find a maximum induced $\gclass$-subgraph which can be obtained from the guessed subset by adding vertices of $K$. This takes time $\cOs(2^{\kappa_0\cdot |K|})$ for each branch. Thus the running time in this case is $\cOs(2^{|V\setminus K|}\cdot 2^{\kappa_0\cdot |K|} )\leq \cOs(2^{(1-\alpha)n}\cdot 2^{\kappa_0\cdot \alpha n} )=\cOs(2^{(1-(1-\kappa_0)\alpha)n})$. Note that $1-(1-\kappa_0)\alpha<1$ for $\alpha>0$ and $\kappa_0<1$.

\smallskip\noindent\textbf{Case~B:} \emph{$G$ has no clique of size $\alpha n$.}\smallskip

Firstly, we search for solutions that have at most $n/2 -\beta n$ or at least $n/2 +\beta n$ vertices for some $\beta$, $0<\beta <1/16$, to be determined later. For this, we may apply a simple brute-force check that tries all vertex subsets of size at most $\lceil n/2-\beta n \rceil$ or at least $\lfloor n/2+\beta n \rfloor$ in time $\cOs(\binom{n}{\lceil n/2-\beta n\rceil})$; note that for $\beta>0$ it holds that $\cOs(\binom{n}{\lceil n/2-\beta n\rceil})\leq \cOs(2^{\kappa_0n})$ for some $\kappa_0<1$ depending on $\beta$.

\smallskip\noindent\textbf{Step 2.} \emph{Iterate through all subsets of vertices of size at most $n/2-\beta n$ or at least $n/2+\beta n$, and for each of them check if it induces a graph belonging to $\gclass$. If some subset of size at least $\lceil n/2+\beta n\rceil$ induces a $\gclass$-graph, output the subgraph induced by any of such subsets of maximum cardinality, and terminate the algorithm. If no subset of size exactly $\lfloor n/2-\beta n \rfloor$ induces a $\gclass$-graph, output the subgraph induced by the maximum size subset inducing a $\gclass$-graph among those of size at most $\lfloor n/2-\beta n \rfloor$, and terminate the algorithm.}\smallskip

Correctness of Step~2 is obvious. If execution of Step~2 did not terminate the algorithm, we know that the cardinality of the vertex set of a maximum induced subgraph belonging to $\gclass$ is between $n/2-\beta n$ and $n/2+\beta n$. We proceed to further steps with this assumption.

Let $H$ be a maximum induced $\gclass$-subgraph of $G$. We do not know how $H$ looks like and the only information about $H$ we have so far is that
\begin{itemize}
\item[(i)]
  $H$ has no clique of size $\alpha n$, and \item[(ii)] that $n/2 -\beta n \leq |V(H)|\leq n/2 +\beta n$. 
  \end{itemize}
  Let us note that the number of vertices of $G$ not contained in $H$ is also between $n/2 -\beta n$ and $n/2 +\beta n$. 

We now use Property (4) to find a $\frac{2}{3}$-balanced clique separator in $H$. More precisely, there is a clique $S$ in $H$ such that $V(H)\setminus S$ may be partitioned into sets $X_1$ and $X_2$ such that
\begin{itemize}
\item[(i)] $\frac{1}{3}|V(H)|-|S| \leq |X_1|, |X_2| \leq \frac{2}{3}|V(H)|$, and 
\item[(ii)] there is no edge between $X_1$ and $X_2$ in $G$.
\end{itemize}
As $S$ is also a clique in $G$, we have that $|S|\leq \alpha n$. Therefore, observe that $|X_1|,|X_2|\geq (\frac{1}{6}-\frac{\beta}{3}-\alpha)n > \frac{1}{8} n$, since $\beta<1/16$ and $\alpha<1/48$.  Property (4) gives us more algorithmic properties of the partition $(X_1,S,X_2)$ of $V(H)$; these properties will be useful later. As $\alpha$ is small, we may afford the following branching step.

\smallskip\noindent\textbf{Step 3.} \emph{Branch into at most $(1+\alpha n)\binom{n}{\alpha n}\cdot (n+1)^2$ subproblems, in each  fixing a different subset of $V$ of size at most $\alpha n$ as $S$, as well as the cardinalities of $X_1$, $X_2$. Discard all the branches where $S$ is not a clique.}\smallskip

\newcommand{\Nn}{N'}

From now on we focus on one subproblem; hence, we assume that the clique $S$ is fixed and the cardinalities of $X_1,X_2$ are known. Let $G'=G\setminus S$; to ease the notation, for $X\subseteq V(G')$ we denote $\Nn[X]=N_{G'}[X]$ and $\Nn(X)=N_{G'}(X)$. We now consider two cases of how the structure of the optimal solution $H$ may look like, depending on how many connected components $H\setminus S$ has. The threshold is $\gamma n$ for a small constant $\gamma>0$ to be determined later.

\smallskip\noindent\textbf{Step 4.} \emph{Branch into two subproblems: in the first branch assume that $H\setminus S$ has at most $\gamma n$ connected components, and in the second branch assume that $H\setminus S$ has more than $\gamma n$ connected components.}\smallskip

In the branches of Step~4 the algorithm checks several cases, and for every case proceeds with further branchings. To ease the description, we do not distinguish these branchings as separate Steps, but rather explain them in the text.

\smallskip\noindent\textbf{Branch~B.1:} \emph{Graph $H\setminus S$ has at most  $\gamma n$ connected components.} \smallskip

We first branch into at most $(n+1)^3$ subproblems, in each fixing the cardinalities of sets $\Nn(X_1)$, $\Nn(X_2)$ and $\Nn(X_1)\cap \Nn(X_2)$ such that $|\Nn(X_1)\cap \Nn(X_2)|\leq |\Nn(X_1)|,|\Nn(X_2)|\leq n-(|S|+|X_1|+|X_2|)$. From now on we assume that these cardinalities are fixed. We consider a few cases depending on the sizes of $\Nn(X_1)$, $\Nn(X_2)$ and $\Nn(X_1)\cap \Nn(X_2)$; in these cases we use small constants $\delta,\varepsilon$, to be determined later.

\smallskip\noindent\textbf{Case~B.1.1:} \emph{$||\Nn(X_1)|-|X_1||\geq \delta n$, or $||\Nn(X_2)|-|X_2||\geq \delta n$.}\smallskip

We concentrate only on the subcase of $||\Nn(X_1)|-|X_1||\geq \delta n$, as the second subcase is symmetric.
As the number of components is small, their approximate location can be guessed at a cost of a small running time overhead as follows. Let $P_1$ be a set of vertices of size at most $\gamma n$ that is constructed by picking one vertex from every component of $G[X_1]=H[X_1]$. We branch into at most $(1+\gamma n)\binom{n}{\gamma n}$ subproblems, in each fixing a different subset of size at most $\gamma n$ as $P_1$.

We add an artificial vertex $v_1$ to $G'$, make it adjacent to all the vertices of $P_1$, and enumerate all vertex sets of the new graph that (i) are connected, (ii) contain $P_1\cup \{v_1\}$, (iii) are of size $|X_1|+1$ and have neighbourhood of size $|\Nn(X_1)|$. By Proposition~\ref{le:connectedComp}, the number of such sets is at most $\binom{|X_1|+|\Nn(X_1)|}{|X_1|}$ and they can enumerated in time $\cOs(\binom{|X_1|+|\Nn(X_1)|}{|X_1|})$; note that here we enumerate candidates for such sets using Proposition~\ref{le:connectedComp} for vertex $v_1$, and filter out all the subsets that do not contain $P$. Clearly, $X_1\cup \{v_1\}$ is among the enumerated sets. 

We therefore branch into at most $\binom{|X_1|+|\Nn(X_1)|}{|X_1|}$ subproblems, in each fixing a different set out of the enumerated ones as $X_1$ (after excluding $v_1$). Moreover, in each subproblem we branch further into at most $\binom{n-|X_1|-|\Nn(X_1)|}{|X_2|}$ subproblems, in each fixing a different subset of $V\setminus (\Nn[X_1]\cup S)$ as $X_2$. For each of these subproblems we check whether $G[X_1\cup X_2\cup S]$ belongs to $\gclass$ in polynomial time.

Thus we obtain at most 
$$(1+\gamma n)\cdot \binom{n}{\gamma n} \cdot\binom{|X_1|+|\Nn(X_1)|}{|X_1|} \cdot\binom{n-|X_1|-|\Nn(X_1)|}{|X_2|}$$
subproblems. Since $||\Nn(X_1)|-|X_1||\geq \delta n$, we infer that $\binom{|X_1|+|\Nn(X_1)|}{|X_1|}\leq \cOs(2^{\kappa_1(|X_1|+|\Nn(X_1)|)})$~for some $\kappa_1<1$, depending on $\delta$. On the other hand, \[\binom{n-|X_1|-|\Nn(X_1)|}{|X_2|}\leq \cOs(2^{n-|X_1|-|\Nn(X_1)|}).\] Since $|X_1|+|\Nn(X_1)|\geq |X_1|\geq \frac{1}{8}n$, we have that in this case \[\binom{|X_1|+|\Nn(X_1)|}{|X_1|} \cdot\binom{n-|X_1|-|\Nn(X_1)|}{|X_2|}=\cOs(2^{\kappa_2 n})\] for some $\kappa_2<1$ depending on $\delta$. Hence, the total number of branches produced by Case B.1.1, including the overheads from guessing $S$ and cardinalities, is bounded by $\cOs( \binom{n}{\alpha n}\cdot \binom{n}{\gamma n} \cdot 2^{\kappa_2 n})$. Given $\kappa_2$, which depends on $\delta$ only, we may choose $\alpha$ and $\gamma$ small enough so that this number is smaller than $\cOs(2^{\kappa_3 n})$ for some $\kappa_3<1$.
 
\smallskip\noindent\textbf{Case~B.1.2:} \emph{Case B.1.1 does not apply, but $|\Nn(X_1)\cap \Nn(X_2)|\geq \varepsilon n$.}\smallskip

We proceed similarly to Case~B.1.1, but we change the strategy for guessing the set $X_2$: instead of a brute-force guess in $V\setminus (\Nn[X_1]\cup S)$, we again make use of Proposition~\ref{le:connectedComp}. Let $P_1,P_2$ be sets of vertices of size at most $\gamma n$ that are constructed by picking one vertex from every component of $G[X_1]=H[X_1]$ and $G[X_2]=H[X_2]$, respectively. Similarly as in the previous case, branch into at most $(1+\gamma n)^2\cdot \binom{n}{\gamma n}^2$ subproblems, in each fixing $P_1$ and $P_2$. Again, we enumerate at most $\binom{|X_1|+|\Nn(X_1)|}{|X_1|}$ candidates for $X_1$ by adding an artificial vertex adjacent to all the vertices of $P_1$, and then we branch into a number of subproblems, in each fixing one of these candidates as $X_1$. We terminate all the branches where $P_2$ and $X_1$ are not disjoint, or there is an edge between them. Note that the total number of created subproblems is at most $\binom{|X_1|+|\Nn(X_1)|}{|X_1|}\leq \cOs(2^{2|X_1|+\delta n})$.

Now consider the graph $G''=G\setminus (\Nn[X_1]\cup S)$. Note that $X_2\subseteq V(G'')$ and the neighbourhood of $X_2$ in $G''$ is of size at most $|\Nn(X_2)|-\varepsilon n$, as at least $\varepsilon n$ vertices from the intersection with $\Nn(X_1)$ have been removed. Therefore, we can add an artificial vertex $v_2$ in $G''$ adjacent to all the vertices of $P_2$, and apply Proposition~\ref{le:connectedComp} to it. Similarly as in the case of $X_1$, we can enumerate at most \[\binom{|X_2|+|\Nn(X_2)|-|\Nn(X_1)\cap \Nn(X_2)|}{|X_2|}\] candidates for the set $ X_2$ in time  \[\cOs\left(\binom{|X_2|+|\Nn(X_2)|-|\Nn(X_1)\cap \Nn(X_2)|}{|X_2|}\right).\] Then we branch further into at most $\binom{|X_2|+|\Nn(X_2)|-|\Nn(X_1)\cap \Nn(X_2)|}{|X_2|}$~subproblems, in each fixing one of the candidates as $X_2$. As $|\Nn(X_2)|\leq |X_2|+\delta n$ and $|\Nn(X_1)\cap \Nn(X_2)|\geq \varepsilon n$, we have that $|X_2|+|\Nn(X_2)|-|\Nn(X_1)\cap \Nn(X_2)|\leq 2|X_2|-(\varepsilon-\delta)n$.

When $X_1$ and $X_2$ are fixed, in polynomial time we check whether the graph $G[X_1\cup X_2\cup S]$ belongs to $\gclass$. Observe that $\cOs(2^{2|X_1|+\delta n})\cdot \cOs(2^{2|X_2|-(\varepsilon-\delta) n})=\cOs(2^{2(|X_1|+|X_2|)-(\varepsilon-2\delta)n})$; moreover, $|X_1|+|X_2|\leq n/2 +\beta n$. Hence, given $\varepsilon>0$ we may choose $\delta$ and $\beta$ small enough so that $\cOs(2^{2|X_1|+\delta n})\cdot \cOs(2^{2|X_2|-(\varepsilon-\delta) n})\leq \cOs(2^{\kappa_4 n})$ for some $\kappa_4<1$ depending on $\varepsilon$. Now observe that the total number of branches produced in Case B.1.4, including overheads from guessing $S$, cardinalities, as well as $P_1$ and $P_2$, is in  $\cOs(\binom{n}{\alpha n}\cdot \binom{n}{\gamma n}^2)\cdot \cOs(2^{\kappa_4 n})$, so given $\kappa_4$ we may choose $\alpha$ and $\gamma$ small enough so that the total number of subbranches produced is at most $\cOs(2^{\kappa_5 n})$ for some $\kappa_5<1$. Every subbranch is then processed in polynomial time.

\smallskip\noindent\textbf{Case~B.1.3:} \emph{None of the cases B.1.1 or B.1.2 applies.}\smallskip

Summarizing, sets $X_1$ and $X_2$ have the following properties:
\begin{itemize}
\item $\frac{1}{6}n - \frac{\beta}{3} n-\alpha n \leq |X_1|,|X_2| \leq \frac{1}{3}n +\frac{2\beta}{3} n$,
\item $\frac{1}{2}n-(\alpha+\beta) n\leq |X_1|+|X_2|\leq \frac{1}{2}n+\beta n$,
\item $||\Nn(X_i)|-|X_i||\leq \delta n$ for $i=1,2$, and $|\Nn[X_1]\cap \Nn[X_2]|\leq \varepsilon n$.
\end{itemize}

\newcommand{\both}{{\textrm{both}}}
\newcommand{\none}{{\textrm{none}}}

Let $U_{\both}=\Nn[X_1]\cap \Nn[X_2]=\Nn(X_1)\cap \Nn(X_2)$, $U_{\none}=V(G')\setminus (\Nn[X_1]\cup \Nn[X_2])$, and $U=U_\both\cup U_\none$. We already know that $|U_\both|\leq \varepsilon n$. We now claim that $|U_\none|\leq \zeta n$, where $\zeta=2\alpha+2\beta+2\delta+\varepsilon$. Indeed, we have that
\begin{eqnarray*}
|U_\none| & = & |V(G')|-|X_1|-|X_2|-|\Nn(X_1)|-|\Nn(X_2)|+|\Nn(X_1)\cap \Nn(X_2)|\\
& \leq & n-2(|X_1|+|X_2|)+2\delta n+\varepsilon n \leq (2\alpha+2\beta+2\delta+\varepsilon) n
\end{eqnarray*}

\begin{figure}[htbp!]
                \centering
                \def\svgwidth{0.6\columnwidth}
                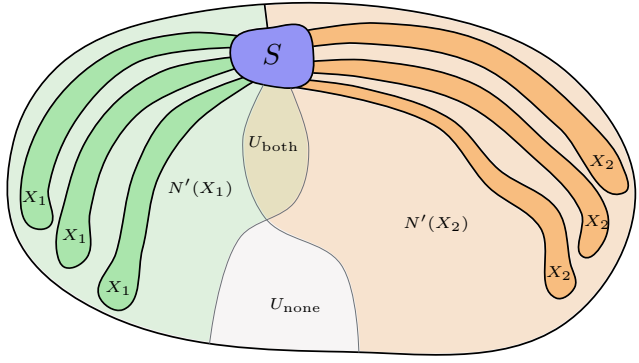
\caption{Situation in Case B.1.3. Neighbourhoods $\Nn(X_1)$ and $\Nn(X_2)$ have almost equal sizes to $X_1$ and $X_2$, respectively, while $U_\both$ and $U_\none$ contain only a tiny fraction of vertices.}\label{fig:large_comp}
\end{figure}

Given that sets $U_\both$ and $U_\none$ are small, we may fix them with $\cOs(\binom{n}{\varepsilon n}\cdot\binom{n}{\zeta n})$ overhead in the running time: we branch into $\cOs(\binom{n}{\varepsilon n}\cdot \binom{n}{\zeta n})$ subproblems, in each fixing a pair of disjoint subsets of $V\setminus S$ of cardinalities at most $\varepsilon n$ and $\zeta n$ as $U_\both$ and $U_\none$, respectively. Note that then $V(G')\setminus U$ is the symmetric difference of $\Nn[X_1]$ and $\Nn[X_2]$; let $I=V(G')\setminus U$. We are left with determining which part of $I$ is in $X_1\cup X_2$, and which is outside.

Observe that every vertex of $I$ is in exactly one of the two sets: $N[X_1]$ or $N[X_2]$. Hence, by Property (4) of $\gclass$, we may look for subsets $X_1,X_2$ of $I$, such that (i) algorithm $\alg$ run on $G[X_1\cup S]$ and $G[X_2\cup S]$ with clique $S$ distinguished provides a positive answer in both of the cases, and (ii) $I$ is a disjoint union of $N[X_1]$ and $N[X_2]$. We model this situation as an instance of the \twotable problem as follows. For $i=1,2$, enumerate all the subsets of $I$ of size $|X_i|$ as candidates for $X_i$, and discard all the candidates for which the algorithm $\alg$ does not provide a positive answer when run on the subgraph induced by the candidate plus the clique $S$. For each remaining candidate subset create a binary vector of length $|I|$ indicating which vertices of $I$ belong to its closed neighbourhood. Construct matrices $T_1,T_2$ by putting the vectors created for candidates for $X_1,X_2$ as columns of $T_1,T_2$, respectively. Now, we need to check whether one can find a column of $T_1$ and a column of $T_2$ that sum up to a vector consisting only of ones. 

As $|X_i| \leq \frac{1}{3}n +\frac{2\beta}{3} n$ for $i=1,2$, we have that tables $T_1,T_2$ have at most $\binom{n}{\frac{1}{3}n +\frac{2\beta}{3} n}$ columns, which is $\cOs(2^{\kappa_6 n})$ for some universal constant $\kappa_6<1$ (recall that $\beta<1/16$, so $\frac{1}{3}n +\frac{2\beta}{3} n<\frac{3}{8}n$). Hence, by Proposition~\ref{lem:2table} we may solve the obtained instance of \twotable in $\cOs(2^{\kappa_6 n})$ time. The total running time used by Case B.1.3, including the overheads for guessing clique $S$, set $U$ and cardinalities, is $\cOs(\binom{n}{\alpha n}\cdot \binom{n}{\varepsilon n}\cdot \binom{n}{\zeta n}\cdot 2^{\kappa_6 n})$; note that we may choose $\alpha,\beta,\delta,\varepsilon$ small enough so that this running time is $\cOs(2^{\kappa_7 n})$ for some $\kappa_7<1$.

\smallskip

\smallskip\noindent\textbf{Branch~B.2:} \emph{Graph $H\setminus S$ has more than $\gamma n$ connected components.}\smallskip

Consider connected components of $H\setminus S$ and fix a large constant $L>2$ depending on $\gamma$, to be determined later. We say that a component containing at most $C=L/\gamma$ vertices is \emph{small}, and otherwise it is \emph{large}. Let $r_\ell$ and $r_s$ be the numbers of large and small components of $H\setminus S$, respectively. The number of vertices contained in large components is hence at least $\frac{L\cdot r_\ell}{\gamma}$. Thus, $\frac{L\cdot r_\ell}{\gamma} \leq n$, $r_\ell\leq \frac{\gamma n}{L}$ and, consequently, $r_s\geq \gamma n-r_\ell\geq \gamma n(1-\frac{1}{L})\geq \frac{\gamma n}{2}$. Since small components are nonempty, they contain at least $\frac{\gamma n}{2}$ vertices in total. 

Let us summarize the situation; see Figure~\ref{fig:large_comp} for reference. The vertices of $V$ can be partitioned into disjoint sets $S$, $X$, $N_X$, $Y$, and $Z$, where 
\begin{itemize}
\item[(i)] $S$ is the clique guessed in Step~3; 
\item[(ii)] $X$ are the vertices contained in large components of $H\setminus S$; 
\item[(iii)] $N_X=\Nn(X)$; 
\item[(iv)] $Y$ are the vertices contained in small components of $H\setminus S$; 
\item[(v)] $Z$ consists of vertices not contained in $H$ and not adjacent to $X$.
\end{itemize}
Note that $V(H)=S\cup X \cup Y$. Unfortunately, even given $X$ and $S$, the algorithm still cannot deduce the solution: we still need to split the remaining part $V\setminus (\Nn[X]\cup S)$ into $Y$ that will go into the solution, and $Z$ that will be left out. However, as we know that $G[X]$ has a small number of components, we can proceed with a branching step that guesses $X$ using Proposition~\ref{le:connectedComp}. Let $P$ be a set of vertices that contains one vertex from each connected component of $G[X]$; we have that $|P|=r_\ell\leq \frac{\gamma n}{L}$. 

\begin{figure}[htbp!]\label{fig:large_comp}
                \centering
                \def\svgwidth{0.6\columnwidth}
                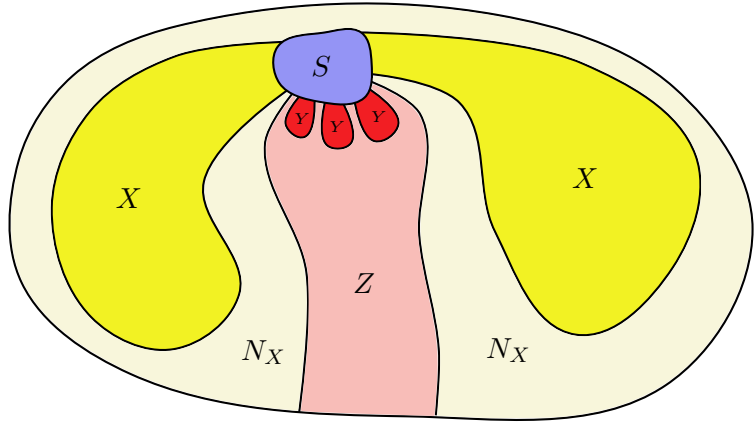
\caption{Situation in Branch B.2. Note that even if we fix $X$ and $S$, the remaining part $V\setminus (\Nn[X]\cup S)$ still needs to be partitioned between $Y$ and $Z$.}\label{fig:large_comp}
\end{figure}

\smallskip\noindent\textbf{Step 5.} \emph{Branch into at most $(n+1)^4$ subbranches fixing $r_\ell,|X|,|Y|,|\Nn[X]|$. Then branch into $\binom{n}{r_\ell}\leq \binom{n}{\frac{\gamma n}{L}}$ cases, in each fixing a different set of size $r_\ell$ as a candidate for $P$. Add an artificial vertex $v_1$ adjacent to $P$, and using Proposition~\ref{le:connectedComp} in $\cOs(\binom{|\Nn[X]|}{|X|})\leq \cOs(2^{|\Nn[X]|})$ time enumerate at most $\binom{|\Nn[X]|}{|X|}\leq 2^{|\Nn[X]|}$ vertex sets that (i) are connected, (ii) contain $P\cup \{v_1\}$, (iii) are of size $|X|+1$ and have neighbourhood of size $|\Nn(X)|$. Note that we can do it by filtering out sets that do not contain $P$ from the list given by Proposition~\ref{le:connectedComp}. As $X\cup \{v\}$ is among enumerated candidates, branch into at most $2^{|\Nn[X]|}$ subcases, in each fixing a different candidate for $X$.}\smallskip

Let $R=G[V\setminus (\Nn[X]\cup S)]$. Note that we need to have $|V(R)|\geq |Y|\geq r_s\geq \frac{\gamma n}{2}$, so if $|V(R)|<\frac{\gamma n}{2}$ then we may safely terminate the branch. We will now use the fact that the input graph does not contain any forbidden induced subgraphs of size bounded by some bound $\ell$; recall that this assumption was justified by an application of Proposition~\ref{lemma:finite_deletion}. We set $\ell=3C^2+1$; hence, whenever we examine an induced subgraph of $G$ of size at most $\ell$, we know that it belongs to $\gclass$. The later steps of the algorithm are encapsulated in the following lemma.

\begin{lemma}\label{le:small-enumeration}
Assuming $\alpha<\frac{\gamma}{104C^3}$ and $\ell=3C^2+1$, there exists a universal constant $\rho<1$ and an algorithm working in $\cOs(2^{\rho |V(R)|})$ time that enumerates at most $\cO(2^{\rho |V(R)|})$ candidate subsets of $V(R)$, such that $Y$ is among the enumerated candidates.
\end{lemma}

Before we proceed to the proof, let us observe that application of Lemma~\ref{le:small-enumeration} finishes the whole algorithm. Indeed, so far in the branching procedure we have an overhead of $\cOs(\binom{n}{\alpha n}\cdot \binom{n}{\frac{\gamma n}{L}}\cdot 2^{|\Nn[X]|})$ for guessing $S$ and $X$. If we now enumerate and examine --- by testing whether $G[X\cup S\cup Y]\in \gclass$ --- all the candidates for $Y$ given by Lemma~\ref{le:small-enumeration}, we arrive at running time \[\cOs(\binom{n}{\alpha n}\cdot \binom{n}{\frac{\gamma n}{L}}\cdot 2^{|\Nn[X]|}\cdot 2^{\rho |V(R)|}).\]

Since $|\Nn[X]|+|V(R)|\leq n$, $\rho<1$ is a universal constant and $|V(R)|\geq \frac{\gamma n}{2}$, given $\gamma>0$ we may choose $L$ to be large enough and $\alpha>0$ to be small enough (and smaller than $\frac{\gamma}{104C^3}$) so that this running time is $\cOs(2^{\kappa_8 n})$ for some $\kappa_8<1$. Here we exploit the fact that $\rho$ does not depend on $\alpha$, $\gamma$ or $L$. Intuitively, what is really happening at this point is that the threshold $C$ for large components depends on $\gamma$ and $L$, and thus the threshold $\ell$ for forbidden induced subgraphs on which we branch a priori using Proposition~\ref{lemma:finite_deletion} depends on $\gamma$ and $L$. This branching, however, is performed outside the current reasoning and we avoid a loop in the definitions of thresholds. 

\newcommand{\Lu}{\mathcal{L}}
\newcommand{\Ls}{\mathcal{L}_{\textrm{small}}}
\newcommand{\Ll}{\mathcal{L}_{\textrm{large}}}

We proceed to the proof of Lemma~\ref{le:small-enumeration}.

\begin{proof}[Proof of Lemma~\ref{le:small-enumeration}]
The initial step is a classical branching algorithm whose goal is to reduce the degrees in $R$. We say that a vertex $v\in V(R)$ is {\emph{heavy}} if $\deg(v)\geq 3C$, and is {\emph{light}} otherwise. The algorithm will produce a number of branches: pairs $(A,D)$, where $A$ is the set of vertices assumed to be contained in the solution, and $D$ is the set of vertices assumed to be excluded from it. Our goal is to get rid of all the heavy vertices, that is, to achieve a situation where all the vertices in $R\setminus D$ are light (where the degrees are counted in $R\setminus D$). The following claim explains all the demanded properties in a formal way.

\begin{mclaim}\label{cl:highdeg-branch}
There exists a universal constant $\sigma<1$ and an algorithm running in time $\cOs(2^{\sigma |V(R)|})$, which outputs a set of pairs $\Lu=\{(A_1,D_1),(A_2,D_2),\ldots,(A_p,D_p)\}$ of disjoints subsets of $V(R)$ with following conditions satisfied:
\begin{itemize}
\item for every pair $(A_i,D_i)$, all vertices of $R\setminus D_i$ are light in $R\setminus D_i$;
\item there is an index $i_0$ such that $A_{i_0}\subseteq Y$ and $D_{i_0}\cap Y=\emptyset$;
\item $\sum_{i=1}^p \phi((A_i,D_i))\leq 2^{\sigma |V(R)|}$, where $\phi$ is a potential function defined as $\phi((A,D))=2^{\sigma |V(R)\setminus (A\cup D)|}$.
\end{itemize}
\end{mclaim}
\begin{proof}[Proof of Claim~\ref{cl:highdeg-branch}]
The algorithm maintains two disjoint sets $A,D$; $A$ is the set of vertices assumed to be contained in the constructed candidate set, while vertices of $D$ are assumed to be excluded from the constructed candidate set. Naturally, we begin with $A=D=\emptyset$. The algorithm stops branching when it finds out that $R\setminus D$ contains only light vertices. Thus, the output of the branching algorithm is a set of leaf branches $(A_i,D_i)$ where $R\setminus D_i$ only contains light vertices. During branching we ensure the property that there is at least output branch $(A_i,D_i)$ such that $A_i\subseteq Y$ and $D_i\cap Y=\emptyset$; to express this property, we will also say that the branching is {\em{correct}}.

The progress of the algorithm is measured by the potential function $\phi((A,D))=2^{\sigma |V(R)\setminus (A\cup D)|}$ for some universal constant $\sigma<1$ to be determined later. At each branching step we will ensure that the sum of potentials in subbranches is at most the potential of the initial branch. As the potential is always at least $1$, we will produce at most $2^{\sigma |V(R)|}$ leaf branches $(A_i,D_i)$ in total, and their total sum of potentials will be at most $2^{\sigma |V(R)|}$. Each branching step will be performed in polynomial time, so the whole branching algorithm runs in $\cOs(2^{\sigma |V(R)|})$ time.

If $R\setminus D$ does not contain any heavy vertex, we terminate the branching procedure and output the current pair $(A,D)$. Otherwise, graph $R\setminus D$ contains some heavy vertex $v$. For simplicity, assume for now that $v\notin A$. As $R[Y]$ has all the connected components of size at most $C$, we infer that if $v\in Y$, then at most a third of the neighbours of $v$ in $R\setminus D$ can belong to $Y$. Hence, we can afford the following branching step. We branch into a number of subcases. In one subcase, $v$ is assigned to $D$. In the other subcases, $v$ is assigned to $A$ and alignment of all the vertices of $N_{R\setminus D}(v)\setminus A$ is guessed in such a manner that $v$ has less than $C$ neighbours in $A$. As the neighbours of $v$ in $R\setminus D$ can be only in $A$ or neither in $A$ nor in $D$, at most a third of neighbours contained in $N_{R\setminus D}(v)\setminus A$ can go in this manner to $A$. Note here that if $N_{R\setminus D}(v)\setminus A$ is empty, this means that $v$ has already at least $C$ neighbours in $A$ and we may safely terminate the branch. The correctness of the presented branching rules follow directly from the fact that all the connected components of $R[Y]$ are of size at most $C$.

The following combinatorial bound will be useful when controlling the behaviour of the potential.

\begin{mfact}\label{cl:third}
If $|M|=n$, then the number of subsets of $M$ of size at most $n/3$ is bounded by $2^{\sigma'\cdot n}$ for some universal constant $\sigma'<1$. 
\end{mfact}

In fact we can choose $2^{\sigma'}=1.89$. By Fact~\ref{cl:third}, in order to prove that the total potential of resulting instances is at most the initial potential, it suffices to check that for $n=|V(R)\setminus (A\cup D)|$ and $m=|N_{R\setminus D}(v)\setminus A|\geq 1$ it holds that
$$2^{\sigma n} \geq 2^{\sigma (n-1)} + 2^{\sigma'\cdot m}\cdot 2^{\sigma(n-1-m)}.$$
This is however equivalent to
$$2^\sigma \geq 1+2^{m(\sigma'-\sigma)}.$$

Let us choose $1>\sigma>\sigma'$ so that $2^{\sigma}(2^{\sigma}-1)\geq 2^{\sigma'}$. Observe that this can be done since function $f(t)=2^t(2^t-1)$ is continuous and strictly increasing in the neighbourhood of $1$, and $f(1)=2>2^{\sigma'}$. Then
$$2^{\sigma} \geq 1+2^{\sigma'-\sigma}\geq 1+2^{m(\sigma'-\sigma)},$$
since $m\geq 1$ and $\sigma'<\sigma$.

In the remaining case when $v\in A$, we simply omit the branch when $v$ is assigned to $D$. Hence, to bound the total potential of obtained subbranches, we need to check that 
$$2^{\sigma n} \geq 2^{\sigma'\cdot m}\cdot 2^{\sigma\cdot (n-m)},$$
which follows from the fact that $\sigma'<\sigma$. This completes the proof of Claim~\ref{cl:highdeg-branch}.
\cqed\end{proof}

We proceed with the proof of Lemma~\ref{le:small-enumeration}. Let
$$\Lu=\{(A_1,D_1),(A_2,D_2),\ldots,(A_p,D_p)\}$$
be the set of pairs produced by Claim~\ref{cl:highdeg-branch}. We know that (i) $\sum_{i=1}^p \phi((A_i,D_i))\leq 2^{\sigma |V(R)|}$, (ii) for every $i$ all the vertices in $R\setminus D_i$ are light, and (iii) there exists an index $i_0$ such that $A_{i_0}\subseteq Y$ and $D_{i_0}\cap Y=\emptyset$. Let $\Ls$ be the subset of $\Lu$ consisting of pairs $(A_i,D_i)$ such that $|A_i|+|D_i|\geq \frac{|V(R)|}{4}$, and $\Ll$ be the subset of remaining instances from $\Lu$. Now, for every pair $(A,D)\in \Lu$ we produce a number of candidates for $Y$. We handle lists $\Ls$ and $\Ll$   differently.

For every pair $(A,D)\in \Ls$ we proceed by brute force. As the final candidates for $Y$, we output all the sets of form $A\cup Y'$, where $Y'$ is a subset of $V(R)\setminus (A\cup D)$. Clearly, if $(A_{i_0},D_{i_0})\in \Ls$, then $Y$ is among the output candidates. We now estimate how many candidates have been output. 

For $(A,D)\in \Ls$, let $m=|V(R)\setminus (A\cup D)|$. Thus, for $(A,D)$ we produce exactly $2^m$ candidates. Since $m\leq \frac{3}{4}|V(R)|$, we have that $2^m=2^{\sigma m}\cdot 2^{(1-\sigma)m}\leq \phi((A,D))\cdot 2^{\frac{3}{4}|V(R)|(1-\sigma)}$. Hence, the total number of candidates produced for $\Ls$ is at most
$$\sum_{(A,D)\in \Ls} \phi((A,D))\cdot 2^{\frac{3}{4}|V(R)|(1-\sigma)}=2^{\sigma |V(R)|+\frac{3}{4}|V(R)|(1-\sigma)}=2^{\frac{3+\sigma}{4}|V(R)|}.$$
Note that $\frac{3+\sigma}{4}<1$ for $\sigma<1$.

We finally proceed to the pairs from $\Ll$. Let $(A,D)\in \Ll$, and let $Q=V(R)\setminus (A\cup D)$. The following claim is the crucial step in our reasoning:

\begin{mclaim}\label{cl:greedy}
If $(A,D)\in \Ll$ is such that $A\subseteq Y$ and $D\cap Y=\emptyset$, then we have that $|Q\cap Y|\geq \frac{2}{3}|Q|$.
\end{mclaim}

In other words, we may safely assume that in the correct branch at least two thirds of the unresolved vertices must remain in the solution. Before we proceed to the proof of Claim~\ref{cl:greedy}, we present how it will be used to finish the whole algorithm of Lemma~\ref{le:small-enumeration}.

For every pair $(A,D)\in \Ll$ we again proceed by brute force, but we take Claim~\ref{cl:greedy} into consideration as well. That is, we output as candidates all sets of form $A_i\cup Y'$, where $Y'$ is a subset of $Q$ of size at least $\frac{2}{3}|Q|$. By applying Fact~\ref{cl:third} to the complement of $Y'$, we infer that the number of produced choices is at most $2^{\sigma' m}\leq 2^{\sigma m}=\phi((A,D))$, where again $m=|Q|$. Thus, the total number of candidates produced in this manner is at most 
$$\sum_{(A,D)\in \Ll} \phi((A,D))\leq 2^{\sigma |V(R)|}\leq 2^{\frac{3+\sigma}{4}|V(R)|}.$$ 
Concluding, the algorithm will produce at most $2\cdot 2^{\frac{3+\sigma}{4}|V(R)|}$ candidates for $Y$: $2^{\frac{3+\sigma}{4}|V(R)|}$ for $\Ls$ and $2^{\frac{3+\sigma}{4}|V(R)|}$ for $\Ll$. Hence we can take $\rho=\frac{3+\sigma}{4}$. Claim~\ref{cl:greedy} ensures that $Y$ will be among the candidates enumerated for $\Ll$ providing that $(A_{i_0},D_{i_0})\in \Ll$, while we have already argued that $Y$ will be among the candidates enumerated for $\Ls$ providing that $(A_{i_0},D_{i_0})\in \Ls$.

We now proceed to the proof of Claim~\ref{cl:greedy}.

\begin{proof}[Proof of Claim~\ref{cl:greedy}]
Assume for the sake of contradiction that $|Q\cap Y|<\frac{2}{3}|Q|$. Then, since $|Q|\geq \frac{3}{4}|V(R)|$, we have that
$$|Q\setminus Y|>\frac{1}{3}|Q|\geq \frac{1}{4}|V(R)|\geq \frac{\gamma n}{8}.$$
We construct a set $T\subseteq Q\setminus Y$ with the following properties:
\begin{itemize}
\item $T$ is independent in $G$;
\item no two vertices of $T$ are adjacent to the same connected component of $R[Y]$;
\item $|T|\geq \frac{\gamma n}{104C^3}$.
\end{itemize}

The construction of $T$ is performed greedily. We iteratively pick to $T$ an unused vertex $v$ of $Q\setminus Y$ and mark the following vertices of $Q\setminus Y$ as used: (i) $v$ itself, (ii) all the neighbours of $v$ in $Q\setminus Y$, and (iii) all the vertices of $Q\setminus Y$ that are adjacent to any component of $R[Y]$ adjacent to $v$. Recall that the degrees in $R\setminus D$ are bounded by $3C$ and $Q\subseteq V(R)\setminus D$, so $v$ can have at most $3C$ neighbours in $Q\setminus Y$. For the same reason, $v$ can be adjacent only to at most $3C$ connected components of $R[Y]$. Each of these components is of size at most $C$, and each vertex contained in any such component can be adjacent to only $3C$ vertices of $Q\setminus Y$. In total, the number of vertices marked as used, including $v$ itself, is at most $1+3C+3C\cdot C\cdot 3C\leq 13C^3$. Hence, we can always find an unused vertex for at least $\frac{|Q\setminus Y|}{13C^3}\geq \frac{\gamma n}{104C^3}$ rounds. From the construction it trivially follows that the constructed $T$ has the first two requested properties.

We now claim that $R[T\cup Y]\in \gclass$. Indeed, from the fact that vertices of $T$ have degree at most $3C$ in $R\setminus D$ we infer that the connected components of $R[T\cup Y]$ need to be of size at most $1+3C^2$. If $R[T\cup Y]\notin \gclass$, then there would be a forbidden induced subgraph from $\F_\gclass$ in $R[T\cup Y]$. As all graphs in $\F_\gclass$ are connected, this subgraph would need to be contained in one of the connected components of $R[T\cup Y]$, and hence would be of size at most $1+3C^2$. However, we assumed that $G$ does not contain any graph from $\F_\gclass$ of size at most $\ell=1+3C^2$, a contradiction. Hence $R[T\cup Y]\in \gclass$.

We conclude the proof with the crucial observation. Define another candidate $H'$ for the optimum solution by taking $H'=G[X\cup Y\cup T]$. In other words, we remove the clique $S$ from the solution $H$, and insert the set $T$ instead. Clearly, $H'$ is a disjoint union of graphs $G[X]$ and $G[Y\cup T]$; as both of these graphs belong to $\gclass$, so does $H'$. Moreover, as $|S|=\alpha n<\frac{\gamma n}{104C^3} \leq |T|$, we have that $|V(H')|>|V(H)|$. This is a contradiction with optimality of $H$.
\cqed\end{proof}

As discussed before, Claim~\ref{cl:greedy} finishes the proof of Lemma~\ref{le:small-enumeration}.
\end{proof}

\section{Summary of the order of choice of constants}\label{app:choice}

In this section we give a short summary of the order of choice of constants. In the following, by running time {\emph{faster than}} $2^n$ we mean running time of form $\cOs(2^{\kappa n})$ for some $\kappa<1$.

We first examine Case B.1.3. In this case, the running time is $\cOs(\binom{n}{\varepsilon n}\cdot\binom{n}{\zeta n})\cdot \cOs(2^{\kappa_6 n})$ for some universal constant $\kappa_6<1$ such that $\binom{n}{\frac{3}{8}n}=\cOs(2^{\kappa_6 n})$. Hence, we can find a positive upper bound $\eps_0>0$ on $\alpha,\beta,\delta,\varepsilon$, such that choosing these constants smaller than $\eps_0$ results in Case B.1.3 running faster than $2^n$.

We now proceed with Case B.1.2. We first fix any $\varepsilon>0$ such that $\varepsilon<\eps_0$.  As observed in this case, given $\varepsilon>0$ we can find a positive upper bound $\eps_1<\eps_0$ on $\alpha,\beta,\gamma,\delta$, such that for any choice of $\alpha,\beta,\gamma,\delta$ smaller than $\eps_1$ we obtain running time faster than $2^n$.

We proceed similarly with Case B.1.1. We fix any $\delta>0$ such that $\delta<\eps_1$. Again, as observed in this case, given $\delta>0$ we can find a positive upper bound $\eps_2<\eps_1$ such that choosing $\alpha,\gamma$ to be smaller than $\eps_2$ results in running time faster than $2^n$.

Now we examine Branch B.2. Let us fix any $\gamma>0$ such that $\gamma<\eps_2$. Recall that the running time in this branch was $\cOs(\binom{n}{\alpha n}\cdot \binom{n}{\frac{\gamma n}{L}}\cdot 2^{|\Nn[X]|}\cdot 2^{\rho |V(R)|})$, where $|\Nn[X]|+|V(R)|\leq n$, $|V(R)|\geq \frac{\gamma n}{2}$ and $\rho<1$ is a universal constant. Hence, given $\gamma>0$ we can find a positive lower bound $L_0\geq 2$ on $L$ and a positive upper bound $\eps_3<\eps_2$ on $\alpha$, such that taking any $L>L_0$ and positive $\alpha<\eps_3$ gives us running time faster than $2^n$. We fix any $L>L_0$, and by lowering $\eps_3$ if necessary we ensure that inequality $\eps_3<\frac{\gamma}{104C^3}$ holds, where $C=L/\gamma$. Then we can fix the remaining two constants: we fix $\beta$ to be any positive constant smaller than $\eps_1$ so that Step 2 runs faster than $2^n$, and $\alpha$ to be any positive constant smaller than $\eps_3$. Thus we make sure that in Branches B.1 and B.2 we obtain running time faster than $2^n$.

Since Case A works faster than $2^n$ for any $\alpha>0$, namely in $\cOs(2^{(1-(1-\kappa_0)\alpha)n})$ time for some $\kappa_0<1$ depending only on $\cbound$, we infer that the whole algorithm runs faster than $2^n$. Note however, that we assumed that the algorithm runs on $\F'_\gclass$-free graphs, where $\F'_\gclass$ consists of graphs of $\F_\gclass$ of size at most $\ell$, for $\ell=3C^2+1=3\frac{L^2}{\gamma^2}+1$. Since $\F'_\gclass$ is a finite family of graphs, we can apply Proposition~\ref{lemma:finite_deletion} as described in Section~\ref{sec:algo} before Step 1, and obtain running time faster than $2^n$ for the general problem.


\section{Conclusion}\label{sec:conlusion}
Theorem~\ref{thm:chord_2n} shows that for any class of graphs  $\gclass$  satisfying Properties~(1)--(4), a maximum induced subgraph from  $\gclass$ of an $n$-vertex graph can be found in time $\cOs(2^{\lambda n})$ for some $\lambda<1$. Pipelining Proposition~\ref{lemma:finite_deletion} with Theorem~\ref{thm:chord_2n} shows that we moreover may add any finite family of forbidden subgraphs on top of belonging to $\gclass$. More precisely, we have the following theorem.

\begin{theorem}\label{thm:chord_plus}
Let $\F$ be a finite set of graphs and  $\gclass$ be a class of graphs satisfying Properties (1)--(4).
There exists an algorithm which for a given  $n$-vertex graph $G$, finds a maximum induced   $\F$-free   $\gclass$-graph in $G$  in time $\cOs(2^{\lambda n})$ for some $\lambda<1$, where $\lambda$ depends only on $\cbound$ and $\F$.
\end{theorem}
As  mentioned in the introduction, Theorem~\ref{thm:chord_plus} covers such graph classes as proper interval graphs (claw-free interval graphs), Ptolemaic graphs (chordal and gem-free),  block graphs (chordal and diamond-free), or  proper circular-arc graphs  (chordal, claw-free, and $\bar{S}_3$-free). We refer to \cite{brandstadt1999graph} for the definitions and discussions on these graphs.

In this manner, we hope to provide a new insight into Hypothesis~\ref{conj:main} by considering chordal-like graph classes. So far the research on breaking the $2^n$ barrier for the {\sc{Maximum Induced $\gclass$-Subgraph}} problem concentrated mostly on exploiting sparsity of a graph class, like in~\cite{FominGPR08-On,FominTV11,PilipczukP12}, or thinness in terms of treewidth, like in the metaresult of Fomin et al.~\cite{FominTV13}. In this work we were dealing with graph classes which inherently allow existence of large cliques, and thus a new set of tools was needed. Shortly speaking, the crux of our approach is to use existence of balanced clique separators in chordal graphs to apply the \twotable trick of Schroeppel and Shamir \cite{SchroeppelS81-A}. However, this application needed to be preceeded by a long and technical preparation of the instance at hand.

Clearly, the most important research direction stemming from our work is further investigation of Hypothesis~\ref{conj:main}. Since we believe that the fully general statement might turn out to be either false or very hard to prove, we propose some relaxations that can be more approachable. 

Firstly, following the approach of Fomin et al.~\cite{FominTV13} one could require the graph class $\gclass$ to be moreover definable in some logical formalism, for example in $\text{CMSO}_2$ (monadic second-order logic with modulo predicates and quantification over edge subsets) or $\text{CMSO}_1$ (the same as $\text{CMSO}_2$, but without quantification over edge subsets). It can be easily seen that both chordal and interval graphs are definable in $\text{CMSO}_1$ by testing existence of any of the forbidden induced subgraphs. We have two concrete examples of hereditary, polynomial-time recognizable, and $\text{CMSO}_1$-definable graph classes for which we do not know any algorithm faster than $2^n$:
\begin{itemize}
\item {\bf{Perfect graphs}} can be defined as graphs which do not contain an odd hole nor an odd anti-hole; this result is known as the Strong Perfect Graph Theorem~\cite{strongperfect}. Perfect graphs are hereditary, polynomial-time recognizable~\cite{recogperfect}, and containing an odd hole or an odd anti-hole can be easily expressed in $\text{CMSO}_1$.
\item {\bf{Strongly chordal graphs}} are chordal graphs that moreover exclude $\ell${\em{-suns}} for $\ell\geq 3$ as induced subgraphs; we refer to~\cite{brandstadt1999graph} for a broader discussion of this graph class. They are also hereditary, polynomial-time recognizable~\cite{brandstadt1999graph}, and definable in $\text{CMSO}_1$. The reason why they do not fall under the regime of Theorem~\ref{thm:chord_plus} is that $\ell$-suns contain arbitrary large cliques and thus Property~(2) is not satisfied.
\end{itemize}

Secondly, one could impose some structural properties on the set of forbidden induced subgraphs $\F_\gclass$. One obvious relaxation, already used in Property~(2), is requiring that all the graphs from $\F_\gclass$ are connected, or equivalently that $\gclass$ is closed under taking disjoint union. More restrictions on the graphs from $\F_\gclass$ can be further imposed. For instance, it would be interesting to see if requiring that all the graphs from $\F_\gclass$ have treewidth bounded by some constant could help in breaking the $2^n$ barrier; note that this subsumes both the case of chordal and of interval graphs.

Finally, one could deviate from the precise statement of Hypothesis~\ref{conj:main} and replace the condition of being hereditary with connectivity. That is, we would like to find a maximum induced {\emph{connected}} graph belonging to $\gclass$. Of course, our approach fails since the connectivity requirements are not hereditary, and thus Property~(1) is not satisfied. Say, can a maximum induced \emph{connected} chordal subgraph be found faster than $2^n$?


\begin{thebibliography}{5}
{\small{
\bibitem{brandstadt1999graph}
{\sc A.~Brandst\"{a}dt, V.~Le, and J.~P. Spinrad}, {\em Graph Classes. {A}
  Survey}, SIAM Mon. on Discrete Mathematics and Applications, SIAM,
  Philadelphia, USA, 1999.

\bibitem{recogperfect}
{\sc M.~Chudnovsky, G. Cornu\'{e}jols, X. Liu, P.~D. Seymour and K.~Vu\v{s}kovi\'{c}}, {\em Recognizing Berge graphs},  Combinatorica, 25(2) (2005), pp.143--186.

\bibitem{strongperfect}
{\sc M.~Chudnovsky, N.~Robertson, P.~D. Seymour, and R.~Thomas}, {\em The strong perfect graph theorem}, Annals of Mathematics, 164(1) (2006), pp.~51--229.


\bibitem{FominGKLS10}
{\sc F.~V. Fomin, S.~Gaspers, D.~Kratsch, M.~Liedloff, and S.~Saurabh}, {\em
  Iterative compression and exact algorithms}, Theor. Comput. Sci., 411 (2010),
  pp.~1045--1053.

\bibitem{FominGPR08-On}
{\sc F.~V. Fomin, S.~Gaspers, A.~V. Pyatkin, and I.~Razgon}, {\em On the
  minimum feedback vertex set problem: Exact and enumeration algorithms},
  Algorithmica, 52 (2008), pp.~293--307.

\bibitem{FominKratschbook10}
{\sc F.~V. Fomin and D.~Kratsch}, {\em Exact Exponential Algorithms}, Springer,
  2010.

\bibitem{FominTV11}
{\sc F.~V. Fomin, I.~Todinca, and Y.~Villanger}, {\em Exact algorithm for the
  maximum induced planar subgraph problem}, in ESA 2011, vol.~6942 of LNCS,  pp.~287--298.

\bibitem{FominTV13}
{\sc F.~V. Fomin, I.~Todinca, and Y.~Villanger}, {\em Large induced subgraphs via triangulations and CMSO}, 
CoRR, abs/1309.1559 (2013). To appear in the proceedings of SODA 2014.



\bibitem{Fomin:2010ys}
{\sc F.~V. Fomin and Y.~Villanger}, {\em Finding induced subgraphs via minimal
  triangulations}, in STACS 2010, vol.~5 of  LIPICS, pp.~383--394.

\bibitem{FominV12}
\leavevmode\vrule height 2pt depth -1.6pt width 23pt, {\em Treewidth
  computation and extremal combinatorics}, Combinatorica, 32 (2012),
  pp.~289--308.
%
\bibitem{Gaspers:2008rf}
{\sc S.~Gaspers}, {\em Exponential Time Algorithms: Structures, Measures, and
  Bounds}, PhD thesis, University of Bergen, 2008.

\bibitem{GaspersKL12}
{\sc S.~Gaspers, D.~Kratsch, and M.~Liedloff}, {\em On independent sets and
  bicliques in graphs}, Algorithmica, 62 (2012), pp.~637--658.

\bibitem{Golumbic80}
{\sc M.~C. Golumbic}, {\em Algorithmic Graph Theory and Perfect Graphs},
  Academic Press, New York, 1980.

\bibitem{GuptaRS12}
{\sc S.~Gupta, V.~Raman, and S.~Saurabh}, {\em Maximum $r$-regular induced
  subgraph problem: Fast exponential algorithms and combinatorial bounds}, SIAM
  J. Discrete Math., 26 (2012), pp.~1758--1780.

\bibitem{LekkerkerkerB62}
{\sc C.~G. Lekkerkerker and J.~C. Boland}, {\em Representation of a finite
  graph by a set of intervals on the real line}, Fund. Math., 51 (1962),
  pp.~45--64.

\bibitem{LewisY80}
{\sc J.~M. Lewis and M.~Yannakakis}, {\em The node-deletion problem for
  hereditary properties is {NP}-complete}, J. Comput. Syst. Sci., 20 (1980),
  pp.~219--230.
%
%

\bibitem{PilipczukP12}
{\sc M.~Pilipczuk and M.~Pilipczuk}, {\em Finding a maximum induced degenerate
  subgraph faster than $2^n$}, in IPEC 2012, vol.~7535 of LNCS, pp.~3--12.

\bibitem{RamanSS07}
{\sc V.~Raman, S.~Saurabh, and S.~Sikdar}, {\em Efficient exact algorithms
  through enumerating maximal independent sets and other techniques}, Theory
  Comput. Syst., 41 (2007), pp.~563--587.

\bibitem{Robson86}
{\sc J.~M. Robson}, {\em Algorithms for maximum independent sets}, J.
  Algorithms, 7 (1986), pp.~425--440.
  
  \bibitem{SchroeppelS81-A}
{\sc R.~Schroeppel and A.~Shamir}, {\em A {$T=\cO(2\sp{n/2})$},
  {$S=\cO(2\sp{n/4})$} algorithm for certain {NP}-complete problems}, {SIAM J.
  Comput.}, 10 (1981), pp.~456--464.

}}
\end{thebibliography}
\end{document}